\documentclass{ws-ijfcs}
\usepackage{enumerate}
\usepackage{url}
\usepackage{cite}

\usepackage{float}
\usepackage{amsmath}
\usepackage{paralist}
\usepackage{graphics}
\usepackage{epsfig}
\usepackage{graphicx}
\usepackage{epstopdf}
\usepackage[colorlinks=true]{hyperref}

\usepackage{mathtools, amssymb}

\usepackage{mathrsfs}
\usepackage{url}

\usepackage[colorinlistoftodos]{todonotes}

\newcommand{\la}{\langle}
\newcommand{\ra}{\rangle}

\usepackage{booktabs}
\usepackage{color}
\newcommand{\changed}[1]{\textcolor{black}{#1}}

\begin{document}

\markboth{Benjamin Russell and Susan Stepney}
{Instructions for Typing Manuscripts $($Paper's Title\/$)$}

%
\catchline{}{}{}{}{}
%

\title{The Geometry of Speed Limiting Resources in Physical Models of Computation}

\author{Benjamin Russell\footnote{Princeton, NJ 08540, USA}}

\address{Department of Chemistry, Princeton University\\
Princeton, NJ 08540, USA\\
\email{br6@princeton.edu}}

\author{Susan Stepney}

\address{Department of Computer Science, University of York\\
York YO10 5DD, UK\\
susan.stepney@york.ac.uk}

\maketitle

\begin{history}
\received{(Day Month Year)}
\accepted{(Day Month Year)}
\comby{(xxxxxxxxxx)}
\end{history}

\begin{abstract}
We study the maximum speed of quantum computation and how it is affected by limitations on physical resources.
We show how the resulting concepts generalize to a broader class of physical models of computation 
\changed{within dynamical systems} and introduce a specific algebraic structure representing these speed limits.
We derive a family of quantum speed limit results in resource-constrained quantum systems with pure states and a finite dimensional state space, by using a geometric method based on right invariant action functionals on $SU(N)$.
We show that when the action functional is bi-invariant, the \changed{minimum time} for implementing \changed{any} quantum gate using a potentially time-dependent Hamiltonian is equal to the \changed{minimum} time when using a constant Hamiltonian, \changed{thus constant Hamiltonians are time optimal for these constraints}.
We give an explicit formula for the time in these cases, in terms of the resource constraint.
We show how our method produces a rich family of speed limit results, of which the generalized Margolus--Levitin theorem and the Mandelstam--Tamm inequality are special cases.
We discuss the broader context of geometric approaches to speed limits in physical computation, including the way geometric approaches to quantum speed limits are a model for physical speed limits to computation arising from a limited resource.
\end{abstract}

\section{Introduction}

As various models of physical computation are explored,
there is interest in determining the ultimate physical limits of such systems  \cite{Lloyd} wherein limits from quantum mechanics, relativity and thermodynamics are presented.
Speed of computation is one important such limit.
There is interest in the speed limit to physical processes \cite{SLPP}, speed limits in open quantum systems \cite{SLOS} and speed limits in non-Markovian systems \cite{SLNM}.
Specifically, there is much interest in the speed limit to quantum information processing (QIP) tasks; a range of perspectives can be found in \cite{ACAR, NFINQSL, MLB, OCQSL, QCOCU}.
There is also recent interest in applications of geometry to time optimal quantum control \cite{me1, me2, Bro1, Bro2, Bro3, ACAR, rage, GGQSL}.

Here we study the optimal time to implement an arbitrary (special unitary) quantum gate in a system with constrained time-dependent Hamiltonian using a generalized version of the geometric technique used in \cite{me1, me2, Bro1, Bro2, Bro3}.
As well as being important in its own right, this problem serves as a test-bed for applying the same geometric technique to assessing the limits of physical models of computation more broadly.

In computer science, it is typical to assess an algorithm by understanding its space and time requirements.
We present a methodology for mathematically relating a physical resource limitation to the optimal time to implement a computation under that limitation in a large class of physical models of computation.
In the quantum case, in the absence of a constraint on the Hamiltonian, there would be no speed limit to the implementation of any quantum gate.
Here we study a general type of constraint.
A suitably general, time-independent constraint on a time-dependent Hamiltonian for an $N$-level quantum system (such as those typical in quantum control and quantum computation) can be represented as a function 
$F:\mathfrak{su}(N) \rightarrow \mathbb{R}$,
where the allowed set of Hamiltonians is represented by imposing $F(-i\hat{H}_t) = 1$ for all time.
\changed{(Here $\mathfrak{su}(N)$ is the special unitary Lie algebra, the tangent space at the identity $T_{I}SU(N)$ on the group $SU(N)$ of special unitary matrices.)}
That is, only Hamiltonians which are on the unit level set of $F$ are admissible.
This class of constraints contains many physically familiar ones, including $|| H_t || = 1$ for any matrix norm $|| \cdot ||$.
If $F$ were to represent a time \emph{dependent} constraint, then the function $F$ would itself have to be time dependent.
This scenario is interesting as it represents the case of a resource upon which the limitation varies in time; investigating this idea further will form the basis of further work but will not be explored further herein.

This kind of resource constraint is an alternative and \changed{complementary} to the `forbidden directions' type constraints studied in \cite{me2}.
Informally, the level set of a smooth function $F$ on $\mathfrak{su}(n)$ is generally a smooth manifold (perhaps with singularities) of dimension one less than the \changed{dimension} of $\mathfrak{su}(n)$ (i.e. codimension 1).
As such, it is not possible to represent the \changed{constraint} that the Hamiltonian can only be drawn from a vector space of traceless Hermitian matrices which is of dimension less than the maximal one ($n^2-1$ for $n$ levels).

One class of constraints that is particularly amenable to geometric analysis is those that are degree $1$ positive homogeneous (PH), that is, $F(\lambda A) = \lambda F(A), \ \ \forall \lambda > 0$.
One notes that, a PH function is not necessarily a positive {\it function}; the positivity refers to $\lambda$, not $F$.
All norms are \emph{absolutely} homogeneous by definition.
As all absolutely homogeneous functions are also PH, the results derived below for PH functions apply to all matrix norms.

PH functions have a favorable property that is exploited throughout their use in geometry:
the action \changed{function of a given curve} corresponding to a PH (point-wise on each tangent space of a manifold) function does not depend on the parametrization of that curve.
This class of PH constraints is very large, and contains many of those standardly studied and of those that arise naturally in practical contexts.

\section{Right Invariant Actions and Evolution Times}

Given a PH function on $\mathfrak{su}(N)$ one can define a right invariant action functional on the group $SU(N)$ by right extension \cite{bump}.
Let $F$ be a PH function on $\mathfrak{su}(N)$ and define its right extension, also known as the canonical lift of right translation from $SU(N)$ to $TSU(N)$,
$F_{\hat{U}}: T_{\hat{U}}SU(N) \rightarrow \mathbb{R}$ by: $F_{\hat{U}}(\hat{A}) := F(\hat{A} \hat{U}^{\dagger})$.
Given any right invariant PH function on $TSU(N)$ (the tangent bundle to the group) and assuming that $\hat{U}_t$ solves the Schr{\"o}dinger equation: $\frac{d \hat{U}_t}{dt} = -i\hat{H}_t \hat{U}_t$, it is possible to define an action functional $S$ for curves $\hat{U}_t$ on $SU(N)$:
\begin{align}
\label{act}
 S[\hat{U}_t] = \int_{0}^{T} F_{\hat{U}_t} \left(\frac{d \hat{U}_t}{dt} \right) dt =
 \int_{0}^{T} F_{U_t} \left(-i\hat{H}_t \hat{U}_t \right) dt = 
 \int_{0}^{T} F \left(-i\hat{H}_t \right) dt
\end{align}
This action functional is itself right invariant (\changed{$S[U_t V] = S[U_t], \ \ \forall V \in SU(n)$}), which is the crucial property exploited in this derivation.
If the Hamiltonian is constrained such that $F(-i\hat{H}_t) = \kappa$ (a constant) throughout an evolution then $S[\hat{U}_t] = \kappa T$ where $T$ is the time interval of the evolution.
As such, the time optimal trajectories of a system constrained as above will in general be the \changed{action} minimizing curves on $SU(n)$ of such functionals emanating from the group identity.

\section{Action of a Time Independent Trajectory with a Given Endpoint}

Given any PH function $F: \mathfrak{su}(N) \rightarrow \mathbb{R}$, then any time independent, finite dimensional quantum system with Hamiltonian $F(-i\hat{H}) = \kappa$ such that $\hat{U}_T = \hat{O}$ (for some desired gate $\hat{O}$) satisfies:
\begin{align}
\label{tiqsl}
 T = \frac{1}{\kappa} F\left(\log(\hat{O}) \right)
\end{align}
This can be obtained straightforwardly from the exponential form of the time-independent time evolution operator and by taking matrix logs of both sides of $\hat{U}_T = \hat{O}$ and applying $F$.

\section{Examples and Known Results}

There are two well-known limits to the minimum time for a quantum state to transition to an orthogonal state in terms of $\bar{E}$, the system's energy expectation \cite{MLB} and the energy uncertainty $\Delta E$ in the energy uncertainty relation \cite{MT}.
For a discussion of the subtleties of the time-energy uncertainty relation and the varied use of the term see \cite{PB}.
In what follows, $E_0$ refers to the ground state energy of the Hamiltonian.

The task of orthogonalizing any state has its time optimal trajectory residing entirely within a single two (complex) dimensional subspace of state space \cite{Brody}.
As such, it is sufficient to study effectively two level systems when considering the speed limit for this particular process.
In any finite dimensional system we can produce exact orthogonality times by setting:
\begin{align}
\label{swap}
\hat{L} & = 
e^{i \pi/2} \begin{pmatrix}
0 & e^{-i \theta} \\
e^{i \theta} & 0
\end{pmatrix} \\ \nonumber
\hat{O} & = \hat{L} \oplus \hat{I}
\end{align}
as this gate maps $|0\rangle$ to $|1\rangle$ (up to a phase) when $\theta = \pi$.
Conjugating this gate by a unitary matrix $\hat{V}$ results in another gate $\hat{V} \hat{O} \hat{V}^{\dagger}$ which maps $\hat{V}|0\rangle$ to $\hat{V}|1\rangle$ similarly.
One readily checks that $\log(\hat{O}) = \frac{\pi}{2} \hat{L} \oplus \hat{Z}$ where the logarithm chosen is the principal logarithm and $\hat{Z}$ is the zero matrix of the appropriate size.

\changed{In the following examples, all functions named as some labeled $G$, are special cases of the fully general PH function $F$ used throughout this work.}
The PH function (with $p>0$):
\begin{align}
 G^{(|\psi \ra)}_{p}(-i\hat{H}) = \frac{\left(\la \psi | (\hat{H} - E_0 \hat{I})^p | \psi \ra\right)^{1/p}}{\la \psi | \psi \ra}
\end{align}
yields the known results \cite{Zych} generalizing the Margolus--Levitin theorem when applied to eqn(\ref{tiqsl}), gate $\hat{V}\hat{O}\hat{V}^{\dagger}$ and the state $\hat{V} |0\ra$ which $\hat{V}\hat{O}\hat{V}^{\dagger}$ orthogonalizes.
After some tedious algebra, this results in:
\begin{align}\label{eqn:ML}
T = & \frac{1}{\kappa} G^{(\hat{V}| 0 \ra)}_{p} \left(\log \left(\hat{V} \left(\hat{O} \oplus \hat{Z} \right) \hat{V}^{\dagger} \right) \right) \\ \nonumber
  = & \frac{\pi}{2 \kappa} G^{(\hat{V}| 0 \ra)}_{p} \left(\hat{V} \left(\log(\hat{O}) \oplus \hat{Z}\right) \hat{V}^{\dagger}\right) \\ \nonumber
= & \frac{\pi}{2^{1/p}\, \kappa} 
= \frac{\pi}{2^{1/p} ({\large\strut}\bar{E} - E_0)}
\end{align}
which is exactly the (saturated) bound of \cite{Zych} and the Margolus--Levitin bound for $p=1$.

The PH function
\begin{align}
 G^{(|\psi \ra)}(-i\hat{H}) = \frac{\left(\la \psi | (\hat{H} - \bar{E} \hat{I})^2 | \psi \ra\right)^{1/2}}{\la \psi | \psi \ra}
\end{align}
yields, by the same method, the time that saturates the Mandelstam--Tamm inequality \cite{MT}:
\begin{align}\label{eqn:MT}
    T = \frac{\pi}{2 \Delta E}
\end{align}

So both these standard results are special cases of our approach, derived from a specific choice of PH function.
We now consider further functions and derive new results in this family.

Consider the operator norm $||\cdot||_{\text{op}}$, which has previously been employed to analyze the quantum speed limit (QSL), in the context of open systems described by a Lindblad operator \cite{Lutz}, and more generally \cite{DL}.
This norm is equal to the largest singular value of a matrix.
The PH function:
\begin{align}
 G_{\text{op}}(-i\hat{H}) = ||\hat{H} - E_0 \hat{I} ||_{\text{op}} = E_{\text{max}} - E_0
\end{align}
leads to the time:
\begin{align}\label{eqn:norm}
 T = \frac{\pi}{E_{max} - E_0}
\end{align}
The common factor of $\pi$ in the times given in equations~(\ref{eqn:ML}), (\ref{eqn:MT}) and (\ref{eqn:norm}) arises from the matrix logarithm of $\hat{O}$.
By applying an arbitrary PH function to the orthogonalizing gate we find:
\begin{align}
    T_{\text{opt}} = \min_{V \in SU(n)} \frac{\pi F\left( \log( \hat{V} \hat{O}\hat{V}^{\dagger})\right) }{2 \kappa}
\end{align}
In the case that $F$ is $Ad$ invariant \changed{(see thm.\ref{adinv} below for a definition)} this yields:
\begin{align}
    T = \frac{\pi F\left( \log(\hat{O})\right)}{2 \kappa}
\end{align}


In \changed{general}, a gate $\hat{O}$ has more than one matrix logarithm \cite{matlog}. 
In order to obtain the physically optimal time to implement a gate, one must minimize over all logarithms.
However, if $F$ is a monotonically increasing function of the modulus of the eigenvalues of $\log(\hat{O})$, this minimum is always achieved by the principal logarithm.
This case includes all unitarily invariant matrix norms.
\cite{ramm} applies the singular values of the Hamiltonian to a QSL problem to compare both Hermitian and non-Hermitian quantum mechanics \changed{\cite{nonhamqm} (where the speed limit is shown to behave differently \cite{fasternhqm})}.

\section{Constraints for which Time-Independent Trajectories are Optimal for all Gates}\label{sec:time-indep-opt}

The method so far does not involve any optimization (other than over the multivalued the matrix logarithm); it gives a formula for \emph{the} time to implement a gate using a specific trajectory.

As shown in \cite{Brody}, in the case of the Margolis--Levitin theorem, ultimately such an analysis cannot yield such a strong result if only the time-independent case is considered,
because the Mandelstam--Tamm inequality is the `sharpest' possible speed limit for the process of orthogonalizing a state.
Nevertheless, other speed limit formulas still have physical relevance, even in the time-independent case.
Consider an experimenter who knows only the energy expectation $\bar{E}$, or more generally $F(-i\hat{H})$, and does not know the energy uncertainty: to them a bound in terms of other physical quantities is of interest.
As such, PH functions can be considered as representing a type of resource for time optimally implementing a quantum process, and eqn(\ref{tiqsl}) can be read as the time to implement a gate $\hat{O}$ when only $\kappa$ amount of resource $F$ is available.

The action of a curve $S[\hat{U}_t]$ is equal to the time to traverse that curve, so in a system with $F(-i\hat{H}_t)=1$ for all $t$, finding the critical curves (with fixed end points $\hat{I}$ and $\hat{O}$) of such an action is a method for finding the time optimal trajectories for a gate $\hat{O}$. This is a generalization of the Finsler geodesic based method of \cite{me1,me2} for finding optimal times in the presence of constraints.
This method applies to any curve, and thus to systems with time-dependent constrained Hamiltonians.
Here we give a condition for the time-independent trajectory to be the optimal one.

The result relies on a theorem about \emph{Relative Equilibria} \cite{Cram} (also referred to as \emph{geodesic vectors} when the Lagrangian is given by a Riemannian/Finsler metric).
A relative equilibrium is a curve which is simultaneously a stationary curve of a Lagrangian and a one parameter subgroup.
In the case of quantum mechanics and $SU(N)$, one parameter subgroups (of the form $e^{it\hat{A}}$ for some Hermitian operator $\hat{A}$) are  the time-independent trajectories.
The following theorem shows when all the optimal trajectories for a constrained 
system are given by the time-independent trajectories.
\begin{theorem}
The time optimal trajectories to implement any given arbitrary gate $\hat{O} \in SU(N)$ in a quantum system constrained such that $F(-i\hat{H}_t)= \kappa$ (for all time) are the trajectories achieved using a time independent Hamiltonian if and only if $F$ is $Ad$ invariant,
that is, when $F(-i\hat{H}) = F(\hat{V}(-i\hat{H})\hat{V}^{\dagger})$ for all $\hat{V} \in SU(N)$.
\label{adinv}
\end{theorem}
\begin{proof}
We apply a result of \cite[\S 6.1]{Cram}.
The relevant theorem states that any bi-invariant Lagrangian on a compact, connected Lie group has the one-parameter subgroups as its critical curves.
\changed{Here, bi-invariant means, invariant under left, and right multiplication by any two \emph{potentially different} group elements \cite{Alexandrino2015}.
This is in contract to $Ad$ invariance, which only requires invariance under conjugation by a single group element.}
As in eqn.(\ref{act}), set the Lagrangian to be the right extension of $F$:
\begin{align}
\mathcal{L}_{U}(\hat{A}\hat{U}) = F(\hat{A}).
\end{align}
Applying the referenced theorem to this Lagrangian yields the result.
We also note that the forward direction of the proof can be easily obtained by an application of the maximum principle.
\end{proof}
$Ad$-invariant norms on $\mathfrak{su}(N)$ are in one-to-one correspondence with bi-invariant Finsler metrics on $SU(N)$ \cite{deng}.
So it follows that \emph{the} time required to implement a gate with a time-independent system is equal to the optimum time limit (over all time-dependent Hamiltonians) in exactly the cases when the PH function representing the constraint is $Ad$ invariant.
This result generalizes the result of \cite{me2} from constraints represented by a Finsler metric, to an arbitrary PH function on $\mathfrak{su}(N)$.

The above mentioned example PH functions\changed{, in eqs. (4) and (6),} are \emph{not} bi-invariant, except the operator norm.
Some interesting examples of bi-invariant $F$ are the unitarily invariant norms \cite{KYF}.

\section{Gates for which Time Independent Trajectories are Optimal for a Given Constraint}

Where $F$ is not bi-invariant there may still be gates for which the optimal trajectory is achieved with a time-independent Hamiltonian.
The action of eqn.(\ref{act}) associated with such an $F$ is still right invariant by construction.
In this case the associated action is a right invariant Finsler metric on $SU(N)$ and the optimal trajectory is a geodesic \cite{me2}.
It is possible to write a criterion for a gate $\hat{O}$ to be implemented in the case that $F$ is a norm
\cite{HOMGEO, HOMEGEO2}.
By applying \cite[thm 3.1]{HOMGEO} in the special case that the state space is the Lie group (rather than a more general homogeneous space),
the condition for $\hat{X}$ to be a geodesic vector becomes:
\begin{align}
\label{yoyo}
 g_{\hat{X}} (\hat{X} , [\hat{X}, \hat{Z}] ) = 0,\ \ \forall \hat{Z} \in \mathfrak{su}(N)
\end{align}
where $g$ is the Hessian (fundamental tensor \cite{ran} restricted to $SU(N)$ in this case) of $F$.
For a specific quantum gate $\hat{O}$, a condition for it to have a time-independent trajectory as its time-optimal trajectory is:
\begin{align}
 \label{thing}
 g_{\log(\hat{O})} (\log(\hat{O}) , [\log(\hat{O}), \hat{Y}] ) = 0, \ \ \forall \hat{Y} \in \mathfrak{su}(N)
\end{align}

One important case where $g$ can be written explicitly is the example of $F$ being a Randers norm \cite{ran}.
For recent applications to quantum time optimal control see \cite{me1, me2, Bro1, Bro2, Bro3}.
Note that a Randers norm can never be bi-invariant, as there are no non-zero bi-invariant one-forms on $SU(N)$:
the adjoint representation of $SU(N)$ acts transitively on $\mathfrak{su}(N)$, any element of $\mathfrak{su}(N)$ can be sent to the kernel of the candidate one-form by the adjoint action of some $\hat{U} \in SU(N)$.
Thus any bi-invariant one form is zero and a bi-invariant Randers metric is Riemannian.

We have not found a way to analyze the case of a general PH function in this approach, as the Hessian is in general not easy to compute.
However, eqn.(\ref{thing}) could be used as numerical check that a certain gate and constrained systems has as time-optimal time-independent control.
Hence there is no need for any {\sc grape}-like \cite{grape} optimization of pulses in such cases. This can serve as a first check before beginning costly numerical optimization procedures.

\section{Novel Speed Limit Formulas From Old, An Algebra Of Speed Limiting Resources}

PH functions can be combined to obtain new such functions, and hence speed limits. 
For example, 
given any two PH functions $F_1,F_2$, then the following are also PH functions:
\begin{itemize}
\item $F_1+F_2$
\item $F(i\hat{H}) = \left(F_1(i\hat{H})^{p}F_2(i\hat{H}_t)^p \right)^{1/p}$ for any $p>0$
\item $F_{\text{max}}(-i\hat{H}) = \max\{F_1(i\hat{H}), F_2(i\hat{H})\}$
\item $F_{\text{min}}(-i\hat{H}) = \min\{F_1(i\hat{H}), F_2(i\hat{H})\}$
\end{itemize}
These last possibilities, of creating a new action functional from the max/min of two others, leads to bounds of the type discussed in \cite{UBT, QLDE}.
The `unified bound' described in \cite{UBT} can be obtained from the action arising from PH functions: the minimum of the energy expectation and the energy uncertainty.
This unified bound \cite{UBT} is subject to a similar criticism as that given in the discussion of the Margolis--Levitin theorem in \cite{Brody}; we make the same observation as given above (\S\ref{sec:time-indep-opt}) about which quantities are {\it a priori} known to experimenters.

The operations above that combine PH functions into new ones have an important property: combining two $Ad$ invariant PH functions yields a new one, which will also have the time-independent trajectories as its corresponding optimal trajectories.
Just as in \cite{me2}, when the optimal trajectory is not a time-independent one, it can be obtained (\changed{at least} numerically) as a solution to the Euler--Poincar\'e equation corresponding to $F$.

\section{Speed Limits in Resource-Limited Physical Models of Computation}

The case of quantum computation implemented via quantum control with a limited resource is a special case of a more general framework for describing computation in physical systems using control theory.
In the quantum scenario, a limitation on the Hamiltonian is described by the restriction $F(iH)=\kappa$.
In a more general setting, which includes many stochastic control scenarios and Hamiltonian mechanics, where the states of a physical system with some degree of external control are points on a manifold $M$, one can ask the informal question: `what is the least time needed to transition from a state $q_0 \in M$ to a state $q_1 \in M$?'
In this scenario, we associate the states of $M$, or subsets of states, with states of a computation in progress in the manner of the representation relation described in \cite{wdpc}.
With this perspective in mind, one sees that this methodology has application to assessing the minimum time required to implement a computation in a physical device.

The scenario of a speed limiting constraint in a controlled dynamical system can be visualized as in figure (\ref{slf}).
\begin{figure}[tp]
 \centering
 \includegraphics[width=0.7\textwidth]{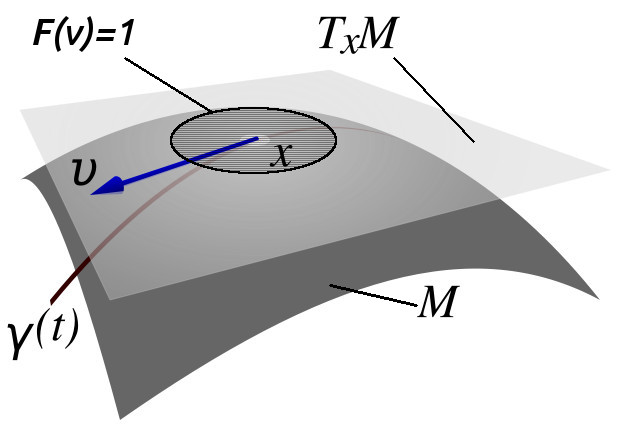}
 \label{slf}
  \caption{A speed limiting constraint in a dynamical system.
   $M$ is the state space of a dynamical system; $T_x M$ is the tangent space at $x$; $\gamma(t)$ is a trajectory; $v$ is the tangent vector to the trajectory at $x$; $F$ represents a speed limiting constraint.
  }
\end{figure}
In this scenario of a constrained navigation problem, the action of any curve $\gamma(t)$ on $M$ connecting $q_0 \in M$ to $q_1 \in M$ in time $T$, subject to the constraint $F_{\gamma(t)}\left(\frac{d \gamma(t)}{dt}\right) = \kappa, \forall t \in [0,T]$, is:
\begin{align}
L[\gamma(t)] & = \int_0^{T} F_{\gamma(t)}\left(\frac{d \gamma(t)}{dt}\right) dt = \kappa T
\end{align}
Provided that $F_{\gamma(t)}$ is a PH function on each tangent space $T_{\gamma(t)}M$, this action is invariant under all positive reparameterizations of the curve $\gamma$.
As such this type of action represents something intrinsically geometrical, and supports \changed{our} claim that physical speed limits in constrained systems, computational or otherwise, are geometrically intrinsic and should be studied using \changed{the framework of geometric time optimal control}.

This observation justifies the idea that physical speed limits to computation are in correspondence with PH functions on the tangent spaces to the state space of the underlying dynamical system enacting the computation.
This statement supports the intuition that \changed{an underlying geometric} structure similar to a distance measure on physical states should form the basis of calculating optimal times.
As such, we propose that such structures should form the basis for understanding the analogue of \changed{Turing machine time-complexity \cite{toc}} for computations embodied in physical substrates \changed{(modeled as dynamical systems), and further,} that the set of all $PH$ functions on a manifold is a good model for speed limiting resources in \changed{such} models of computation.

\section{Discussion and Conclusion}

We have presented a novel mathematical technique for representing a wide class of physical constraints on a quantum Hamiltonian for a finite dimensional system which includes both theoretically and practically interesting examples.

There are some physically meaningful constraints that cannot be represented this way, for example: $\bar{E} = \langle \psi_t|\hat{H}_t | \psi_t \rangle = \kappa$ or $(\Delta E)^2 =  \langle \psi | \left( \hat{H}_t - \bar{E}\right)^2 | \psi_t \rangle = \kappa$ through an evolution.
The states that appear in our constraints are time independent even when $\hat{H}_t$ is time dependent.
We expect that eqn.(\ref{yoyo}) will readily generalize to this situation and yield a procedure similar to that used to prove the Mandelstam--Tamm inequality in \cite{mese} (and many other places) for a wider class of physical constraints.

One possible obstacle to this approach is that a `finite form' of each $PH$ function may not be \changed{known}.
In the Mandelstam--Tamm inequality case, the geodesic distances on complex projective space are known exactly to be $\frac{\pi}{2}$ for orthogonal states.
However, the action of the curve of least action, using an arbitrary $PH$ function, connecting two states may not be known or be impossible to write in closed form in cases other than when the action is the Fubini--Study length of a curve.
The Fubini--Study metric is the only unitarily invariant metric on complex projective space, and is also the only (by an identical proof) unitarily invariant action arising from a $PH$ function on the same space.
This observation clarifies the exact sense in which the Mandelstam--Tamm inequality is the `sharpest' bound \cite{Brody} of the process of orthogonalizing a quantum state.
We have shown that the concepts presented generalize to far more general notions of physical computation modeled as controlled dynamical systems with a speed limiting constraint.

The method presented is a generalization of the representation of a constraint on a quantum Hamiltonian presented in \cite{me2}, which analysed the case that the constraint was represented by a Riemannian metric.
Given both:
\begin{itemize}
\item the scope of the relationship between physical speed limits in systems with an appropriate constraint
\item the broad recent interest in the physical limits to computation
\end{itemize}
we propose, for computations physically embodied in dynamical systems, that:
\begin{itemize}
\item PH functions model speed-limiting resources in physical computation
\item the associated actions of trajectories, yield optimal times for a system to traverse a specific trajectory
\item Constrained navigation problems in dynamical systems are the correct tool of assessing optimal times in physical computation and that such optimal times form the analogue of classical time complexity analysis \cite{toc}
\end{itemize}

\changed{It is also of practical interest, especially if the exact quantum gate required for a computational step is not exactly known, what the `slowest gate' is;
that is, the gate for which the minimum time is the largest given a specific constraint.
This time is discussed in detail without constraints in the work \cite{qclcl}.
Seeking the analogous time, for each constraint of the type discussed in this work would be of interest.}

\begin{table}[H]
\centering
\caption{Summary of results}
\label{tabsum}
\begin{tabular}{l|l|l}
$F$ property & $Ad$-invariant & not $Ad$-invariant \\ \hline
PH function &

\begin{tabular}[c]{@{}l@{}} Constant Hamiltonian\\ optimal for all gates \end{tabular}

& 

\begin{tabular}[c]{@{}l@{}} Constant Hamiltonian\\ optimal for gates solving (14) \end{tabular}

\\ \hline
A norm (all also PH) & \begin{tabular}[c]{@{}l@{}} Constant Hamiltonian\\ optimal for all gates \end{tabular}

& 

\begin{tabular}[c]{@{}l@{}} Constant Hamiltonian \\ optimal for gates solving (14) \end{tabular}
\end{tabular}
\end{table}
\changed{Table \ref{tabsum} summarizes the results on speed limits for quantum gates in systems with constraints represented by $PH$ functions on $\mathfrak{su}(N)$.
One clearly sees that the property of the constraint being represented by a norm, rather than by a more general $PH$ function, is not the key determining factor affecting the minimum time or the nature (constant or otherwise) of the time minimizing control scheme.
This indicates that the scope of investigation in to the geometry of the quantum speed limit, and speed limits for computation in more general classes of physical system should be expanded beyond the investigation of this somewhat limited class of constraints. 
We further note that we have not included the case of choosing from a limited set of constant Hamiltonians during a computation (i.e., a piecewise constant control), as is the case in the quantum circuit model.}

\section{Acknowledgments}

We would like to thank Norman Margolus, Raam Uzdin and Kazuyuki Fujii for several comments on and corrections to a previous version of this work.

\bibliographystyle{unsrt}
\bibliography{mybib}

\begin{thebibliography}{10}

\bibitem{Lloyd}
Seth Lloyd.
\newblock Ultimate physical limits to computation.
\newblock {\em Nature}, 406(6799):1047--1054, 1999.

\bibitem{SLPP}
M.~M. Taddei, B.~M. Escher, L.~Davidovich, and R.~L. de~Matos~Filho.
\newblock Quantum speed limit for physical processes.
\newblock {\em Phys. Rev. Lett.}, 110:050402, Jan 2013.

\bibitem{SLOS}
A.~del Campo, I.~L. Egusquiza, M.~B. Plenio, and S.~F. Huelga.
\newblock Quantum speed limits in open system dynamics.
\newblock {\em Phys. Rev. Lett.}, 110:050403, Jan 2013.

\bibitem{SLNM}
Sebastian Deffner and Eric Lutz.
\newblock Quantum speed limit for non-markovian dynamics.
\newblock {\em Phys. Rev. Lett.}, 111:010402, Jul 2013.

\bibitem{ACAR}
Alberto Carlini, Akio Hosoya, Tatsuhiko Koike, and Yosuke Okudaira.
\newblock Time-optimal quantum evolution.
\newblock {\em Phys. Rev. Lett.}, 96:060503, Feb 2006.

\bibitem{NFINQSL}
M.~A. {Nielsen}.
\newblock {A geometric approach to quantum circuit lower bounds}.
\newblock {\em Quant.\ Info.\ Comp}, 6:213--262, 2006.

\bibitem{MLB}
Norman Margolus and Lev~B. Levitin.
\newblock The maximum speed of dynamical evolution.
\newblock {\em Physica D}, 120:188--195, 1998.

\bibitem{OCQSL}
T.~Caneva, M.~Murphy, T.~Calarco, R.~Fazio, S.~Montangero, V.~Giovannetti, and
  G.~E. Santoro.
\newblock Optimal control at the quantum speed limit.
\newblock {\em Phys. Rev. Lett.}, 103:240501, Dec 2009.

\bibitem{QCOCU}
Jos\'e~P. Palao and Ronnie Kosloff.
\newblock Quantum computing by an optimal control algorithm for unitary
  transformations.
\newblock {\em Phys. Rev. Lett.}, 89:188301, 2002.

\bibitem{me1}
Benjamin Russell and Susan Stepney.
\newblock Zermelo navigation and a speed limit to quantum information
  processing.
\newblock {\em Phys. Rev. A}, 90:012303, 2014.

\bibitem{me2}
Benjamin Russell and Susan Stepney.
\newblock Zermelo navigation in the quantum brachistochrone.
\newblock {\em J. Phys. A}, 48(11):115303, 2015.

\bibitem{Bro1}
Dorje~C. Brody and David~M. Meier.
\newblock Elementary solution to the time-independent quantum navigation
  problem.
\newblock {\em J. Phys. A}, 48(5):055302, 2015.

\bibitem{Bro2}
Dorje~C. Brody and David~M. Meier.
\newblock Solution to the quantum zermelo navigation problem.
\newblock {\em Phys. Rev. Lett.}, 114:100502, 2015.

\bibitem{Bro3}
Dorje~C. Brody, Gary~W. Gibbons, and David~M. Meier.
\newblock Time-optimal navigation through quantum wind.
\newblock {\em New Journal of Physics}, 17(3):033048, 2015.

\bibitem{rage}
Xiaoting Wang, Michele Allegra, Kurt Jacobs, Seth Lloyd, Cosmo Lupo, and Masoud
  Mohseni.
\newblock Quantum brachistochrone curves as geodesics: Obtaining accurate
  minimum-time protocols for the control of quantum systems.
\newblock {\em Phys. Rev. Lett.}, 114:170501, Apr 2015.

\bibitem{GGQSL}
Diego~Paiva Pires, Marco Cianciaruso, Lucas~C. C\'eleri, Gerardo Adesso, and
  Diogo~O. Soares-Pinto.
\newblock Generalized geometric quantum speed limits.
\newblock {\em Phys. Rev. X}, 6:021031, Jun 2016.

\bibitem{bump}
D.~Bump.
\newblock {\em Lie Groups}.
\newblock Springer, 2004.

\bibitem{MT}
L.~Mandelstam and Ig. Tamm.
\newblock The uncertainty relation between energy and time in non-relativistic
  quantum mechanics.
\newblock {\em J. Phys. (USSR)}, 9(4):249--254, 1945.

\bibitem{PB}
P.~{Busch}.
\newblock {The Time-Energy Uncertainty Relation}.
\newblock In J.~G. {Muga}, R.~{Sala Mayato}, and I.~L. {Egusquiza}, editors,
  {\em Time in Quantum Mechanics}, page~69, 2002.

\bibitem{Brody}
D.~C. {Brody}.
\newblock {Elementary derivation for passage times}.
\newblock {\em J. Phys. A}, 36:5587--5593, May 2003.

\bibitem{Zych}
B.~{Zieli{\'n}ski} and M.~{Zych}.
\newblock {Generalization of the Margolus-Levitin bound}.
\newblock {\em Phys. Rev. A}, 74(3):034301, 2006.

\bibitem{Lutz}
S.~{Deffner} and E.~{Lutz}.
\newblock {Quantum Speed Limit for Non-Markovian Dynamics}.
\newblock {\em Phys. Rev. Lett.}, 111(1):010402, 2013.

\bibitem{DL}
Daniel~A. Lidar, Paolo Zanardi, and Kaveh Khodjasteh.
\newblock Distance bounds on quantum dynamics.
\newblock {\em Phys. Rev. A}, 78:012308, 2008.

\bibitem{matlog}
Terry~A. Loring.
\newblock Computing a logarithm of a unitary matrix with general spectrum.
\newblock {\em Numerical Linear Algebra with Applications}, 21(6):744--760,
  2014.

\bibitem{ramm}
Raam Uzdin, Uwe G\"unther, Saar Rahav, and Nimrod Moiseyev.
\newblock Time-dependent {H}amiltonians with 100\% evolution speed efficiency.
\newblock {\em J. Phys. A}, 45(41):415304, 2012.

\bibitem{nonhamqm}
Carl~M Bender.
\newblock Making sense of non-hermitian hamiltonians.
\newblock {\em Reports on Progress in Physics}, 70(6):947, 2007.

\bibitem{fasternhqm}
Carl~M. Bender, Dorje~C. Brody, Hugh~F. Jones, and Bernhard~K. Meister.
\newblock Faster than hermitian quantum mechanics.
\newblock {\em Phys. Rev. Lett.}, 98:040403, Jan 2007.

\bibitem{Cram}
M.~Crampin and T.~Mestdag.
\newblock Relative equilibria of {L}agrangian systems with symmetry.
\newblock {\em Journal of Geometry and Physics}, 58(7):874--887, 2008.

\bibitem{Alexandrino2015}
Marcos~M. Alexandrino and Renato~G. Bettiol.
\newblock {\em Lie Groups with Bi-invariant Metrics}, pages 27--47.
\newblock Springer International Publishing, Cham, 2015.

\bibitem{deng}
S.~Deng.
\newblock {\em Homogeneous Finsler Spaces}.
\newblock Springer, 2012.

\bibitem{KYF}
Ky~Fan and A.~J. Hoffman.
\newblock Some metric inequalities in the space of matrices.
\newblock {\em Proc. AMS}, 6(1):111--116, 1955.

\bibitem{HOMGEO}
D.~{Latifi}.
\newblock {Homogeneous geodesics of left invariant Finsler metrics}.
\newblock {\em ArXiv e-prints}, November 2007.

\bibitem{HOMEGEO2}
Parastoo Habibi, Dariush Latifi, and Megerdich Toomanian.
\newblock Homogeneous geodesics and the critical points of the restricted
  finsler function.
\newblock {\em Journal of Contemporary Mathematical Analysis}, 46(1):12--16,
  2011.

\bibitem{ran}
X.~Cheng and Z.~Shen.
\newblock {\em Finsler Geometry: An Approach via Randers Spaces}.
\newblock Springer, 2013.

\bibitem{grape}
Benjamin Rowland and Jonathan~A. Jones.
\newblock Implementing quantum logic gates with gradient ascent pulse
  engineering: principles and practicalities.
\newblock {\em Phil. Trans. Roy. Soc. A}, 370(1976):4636--4650, 2012.

\bibitem{UBT}
L.~B. {Levitin} and T.~{Toffoli}.
\newblock Fundamental limit on the rate of quantum dynamics: The unified bound
  is tight.
\newblock {\em Phys. Rev. Lett.}, 103(16):160502, 2009.

\bibitem{QLDE}
Vittorio Giovannetti, Seth Lloyd, and Lorenzo Maccone.
\newblock Quantum limits to dynamical evolution.
\newblock {\em Phys. Rev. A}, 67:052109, 2003.

\bibitem{wdpc}
Clare Horsman, Susan Stepney, Rob~C. Wagner, and Viv Kendon.
\newblock When does a physical system compute?
\newblock {\em Proceedings of the Royal Society of London A: Mathematical,
  Physical and Engineering Sciences}, 470(2169), 2014.

\bibitem{toc}
J.C. Martin.
\newblock {\em Introduction to languages and the theory of computation}.
\newblock McGraw-Hill series in computer science. McGraw-Hill, 2003.

\bibitem{mese}
Benjamin Russell and Susan Stepney.
\newblock Applications of {F}insler geometry to speed limits to quantum
  information processing.
\newblock {\em International Journal of Foundations of Computer Science},
  25(04):489--505, 2014.

\bibitem{qclcl}
PIERRE DE~FOUQUIERES and SOPHIE~G. SCHIRMER.
\newblock A closer look at quantum control landscapes and their implication for
  control optimization.
\newblock {\em Infinite Dimensional Analysis, Quantum Probability and Related
  Topics}, 16(03):1350021, 2013.

\end{thebibliography}

\end{document}